\documentclass[a4paper, journal, twocolumn]{IEEEtran}
\IEEEoverridecommandlockouts

\usepackage{bm}
\usepackage{amsmath,amssymb,amsfonts,graphicx,amsthm,mathtools,nicefrac,mathrsfs}
\usepackage{graphicx, subfigure}
\usepackage{times}
\usepackage{color}
\usepackage[table]{xcolor}
\usepackage{booktabs}
\usepackage{multirow}
\usepackage{lscape}

\DeclareMathOperator{\E}{\mathbb{E}}
\DeclareMathOperator{\dist}{dist}

\DeclareMathOperator{\rad}{rad}
\DeclareMathOperator{\diam}{diam}

\renewcommand{\Pr}[2][]{\mathbb{P}_{#1} \left\{ #2 \rule{0mm}{3mm}\right\}}
\newcommand{\ip}[2]{\left\langle#1,#2\right\rangle}
\newcommand{\gw}[1]{\omega\left(#1\right)}

\def \CC {\mathcal{C}}

\def \EE {\mathcal{E}}
\def \NN {\mathcal{N}}

\def \TT {\mathcal{T}}

\def \B {\mathbb{B}}
\def \R {\mathbb{R}}
\def \S {\mathbb{S}}

\def \va {\bm{a}}
\def \vb {\bm{b}}
\def \vd {\bm{d}}

\def \vg {\bm{g}}
\def \vh {\bm{h}}
\def \vs {\bm{s}}
\def \vu {\bm{u}}
\def \vv {\bm{v}}
\def \vw {\bm{w}}
\def \vx {\bm{x}}
\def \vy {\bm{y}}
\def \vz {\bm{z}}
\def \vzero {\bm{0}}

\def \mA {\bm{A}}
\def \mI {\bm{I}}

\newtheorem{theorem}{Theorem}
\newtheorem{lemma}{Lemma}
\newtheorem{proposition}{Proposition}

\newtheorem{assumption}{Assumption}

\theoremstyle{definition}

\theoremstyle{remark}

\makeatletter
\newtheorem*{rep@theorem}{\rep@title}
\newcommand{\newreptheorem}[2]{%
	\newenvironment{rep#1}[1]{%
		\def\rep@title{#2 \ref{##1}}%
		\begin{rep@theorem}}%
		{\end{rep@theorem}}}
\makeatother

\newreptheorem{theorem}{Theorem}
\newreptheorem{lemma}{Lemma}
\newreptheorem{corollary}{Corollary}
\newreptheorem{proposition}{Proposition}

\begin{document}
%
\title{Corrupted Sensing with Sub-Gaussian Measurements}
%
%
\author{
\IEEEauthorblockN{Jinchi Chen\IEEEauthorrefmark{1} and Yulong~Liu\IEEEauthorrefmark{1}}

\IEEEauthorblockA{\IEEEauthorrefmark{1}School of Physics, Beijing Institute of Technology, Beijing 100081, China}

\thanks{This work was supported by the National Natural Science Foundation of China under Grant 61301188.}
}

%

\maketitle

\pagestyle{empty}  
\thispagestyle{empty} 

\begin{abstract}
This paper studies the problem of accurately recovering a structured signal from a small number of corrupted sub-Gaussian measurements. We consider three different procedures to reconstruct signal and corruption when different kinds of prior knowledge are available. In each case, we provide conditions for stable signal recovery from structured corruption with added unstructured noise. The key ingredient in our analysis is an extended matrix deviation inequality for isotropic sub-Gaussian matrices.
\end{abstract}


\section{Introduction}
Corrupted sensing concerns the problem of recovering a structured signal from a relatively small number of corrupted measurements
\begin{align}\label{model: observe}
\vy = \bm{\Phi}\vx^{\star} +\vv^{\star} + \vz,
\end{align}
where $\bm{\Phi}\in\R^{m\times n}$ is the sensing matrix with $m \ll n$, $\vx^{\star}\in\R^n$ is the structured signal, $\vv^{\star}\in\R^m$ is the structured corruption, and $\vz \in \R^m$ is the unstructured observation noise. The goal is to estimate $\vx^{\star}$ and $\vv^{\star}$ from given knowledge of $\vy$ and $\bm{\Phi}$.

This problem has received increasing attention recently with many interesting practical applications as well as theoretical consideration. Examples of applications include face recognition \cite{wright2009robust}, subspace clustering \cite{elhamifar2009sparse}, sensor network \cite{haupt2008compressed}, and so on. Examples of theoretical guarantees include sparse signal recovery from sparse corruption \cite{wright2010dense,li2013compressed,nguyen2013exact,nguyen2013robust,pope2013probabilistic,studer2012recovery,studer2014stable} and structured signal recovery from structured corruption \cite{foygel2014corrupted}. It is worth noting that this model \eqref{model: observe} also includes the signal separation (or demixing) problem \cite{mccoy2014sharp} in which $\vv^{\star}$ might actually contain useful information and thus is necessary to be recovered. In particular, if there is no corruption $(\vv^{\star} = \vzero)$, this model \eqref{model: observe} reduces to the standard compressed sensing problem.

Since this problem is generally ill-posed, recovery is possible when both signal and corruption are suitably structured. Let $f(\cdot)$ and $g(\cdot)$ be suitable norms which promote structures for signal and corruption respectively. We consider three different convex optimization approaches to disentangle signal and corruption when different kinds of prior information are available. Specifically, when prior knowledge of either signal $f(\vx^{\star})$ or corruption $g(\vv^{\star})$ is available and the noise level $\delta$ (in terms of $\ell_2$-norm) is known, it is natural to consider the following constrained convex recovery procedures
\begin{align}\label{Constrained_Optimization_I}
\min_{\vx, \vv} ~f(\vx),\quad\text{s.t.~}&g(\vv) \leq g(\vv^{\star}), ~~\|\vy-\bm{\Phi}\vx-\vv\|_2\leq \delta
\end{align}
and
\begin{align}\label{Constrained_Optimization_II}
\min_{\vx, \vv} ~g(\vv),\quad\text{s.t.~}&f(\vx) \leq f(\vx^{\star}), ~~\|\vy-\bm{\Phi}\vx-\vv\|_2\leq \delta.
\end{align}
When only the noise level $\delta$ is known, it is convenient to use the partially penalized convex recovery procedure
\begin{align}\label{Partially_Penalized_Optimization}
\min_{\vx,\vv}~f(\vx)+\lambda\cdot g(\vv),\quad\text{s.t.}\quad\|\vy-\bm{\Phi}\vx-\vv\|_2\leq \delta,
\end{align}
where $\lambda > 0$ is a tradeoff parameter. When there is no prior knowledge available, it is practical to utilize the fully penalized convex recovery procedure
\begin{align}\label{Fully Penalized Optimization}
\min_{\vx,\vv}\frac{1}{2}\|\vy-\bm{\Phi}\vx-\vv\|_2^2+\tau_1\cdot f(\vx)+\tau_2\cdot g(\vv),
\end{align}
where $\tau_1,\tau_2 > 0$ are some tradeoff parameters.

This paper considers the problem of recovering a structured signal from corrupted sub-Gaussian measurements. The contribution of this paper is threefold:
\begin{itemize}
\item [(1):] First, we consider sub-Gaussian measurements in model \eqref{model: observe}. Specifically, we assume that each row $\bm{\Phi}_i$ of the sensing matrix $\bm{\Phi}$ is independent, centered, and sub-Gaussian random vector with
\begin{equation}
\label{model: subgaus property}
\|\bm{\Phi}_{i}\|_{\psi_2}\leq K/\sqrt{m}~~ \textrm{and} ~~ \E\bm{\Phi}_{i}^T\bm{\Phi}_{i} = \mI_n/m,
\end{equation}
where $\|\cdot\|_{\psi_2}$ denotes the sub-Gaussian norm and $\mI_n$ is the $n$-dimensional identity matrix.

\item [(2):] Second, the unstructured noise $\vz$ is assumed to be bounded $(\|\vz\|_2 \leq \delta)$ or be a random vector with independent centered sub-Gaussian entries satisfying
\begin{align} \label{model: noise}
\|\vz_i\|_{\psi_2} \leq L ~~\text{and}~~ \E\vz_i^2 = 1.
\end{align}

\item [(3):] Third, under the above conditions, we establish performance guarantees for all three convex recovery procedures.

\end{itemize}

It is worth noting that in \cite{mccoy2014sharp} only the constrained convex recovery procedures (\eqref{Constrained_Optimization_I} and \eqref{Constrained_Optimization_II}) were considered under random orthogonal measurements ($m=n$) and noise-free case ($\delta = 0$). In \cite{foygel2014corrupted}, both the constrained convex recovery procedures (\eqref{Constrained_Optimization_I} and \eqref{Constrained_Optimization_II}) and the partially penalized convex recovery procedure \eqref{Partially_Penalized_Optimization} were analyzed under Gaussian measurements and bounded noise case. The results in this paper solve a series of open problems in \cite{foygel2014corrupted} (e.g., allowing non-Gaussian measurements and stochastic unstructured noise in model \eqref{model: observe} and analyzing the fully penalized convex recovery procedure \eqref{Fully Penalized Optimization}).

\section{Preliminaries}
In this section, we review some preliminaries that underlie our analysis.

The \emph{subdifferential} of $f$ at $\vx$ is the set of vectors
\begin{equation*}
  \partial f(\vx) = \{ \vu \in \R^n: f(\vx + \vd) \geq f(\vx) + \langle \vu, \vd\rangle ~~ \textrm{for all} ~~\vd \in \R^n \}.
\end{equation*}
The \emph{tangent cone} of $f$ at $\vx$ is defined as the set of descent directions of $f$ at $\vx$
\begin{equation}\label{DefinitionofTangentCone}
    \mathcal{T}_f = \{\vu \in \R^n: f(\vx+t \cdot \vu) \le f(\vx)~~\textrm{for some}~~t>0\}.
\end{equation}

The \emph{Gaussian width} of a subset $\mathcal{C} \subset \R^n$ is defined as
\begin{align*}
\gw{\CC}=\E\sup_{\vu\in\CC}\ip{\vg}{\vu}, ~~\textrm{where}~~\vg\sim\NN(0,\mI_n).
\end{align*}
While the \emph{Gaussian complexity} for a subset $\mathcal{C} \subset \R^n$ is defined as
\begin{align*}
\gamma(\CC)=\E\sup_{\vu\in\CC}|\ip{\vg}{\vu}|, ~~\textrm{where}~~\vg\sim\NN(0,\mI_n).
\end{align*}
These two geometric quantities are closely related, in particular,
\begin{equation}\label{Relation}
\gamma(\mathcal{C}) \leq  2 w(\mathcal{C})+\|\vu\|_2~~~\textrm{for any point}~~~\vu \in \mathcal{C}.
\end{equation}
The \emph{Gaussian squared distance} $\eta^2(\CC)$ of a subset $\CC\subset\R^n$ is defined as
\begin{align*}
	\eta^2(\CC):=\E\inf_{\vu\in\CC}\|\vg-\vu\|_2^2,~~\textrm{where}~~\vg\sim\NN(0,\mI_n).
\end{align*}

A random variable $X$ is called a \emph{sub-Gaussian random variable} if the Orlicz norm
\begin{equation} \label{definitionofsubgaussian}
\|X\|_{\psi_2} = \inf \{t>0: \E \exp(X^2/t^2)\le 2\}
\end{equation}
is finite. The sub-Gaussian norm of $X$, denoted $||X||_{\psi_2}$, is defined to be the smallest $t$ in \eqref{definitionofsubgaussian}. A random vector $\vx \in \R^n$ is called a \emph{sub-Gaussian random vector} if all of its one-dimensional marginals are sub-Gaussian random variables and its $\psi_2$-norm is defined as
\begin{equation} \label{subgaussian vector}
\|\vx\|_{\psi_2}:=\sup_{\vy\in\S^{n-1}}\big\| \ip{\vx}{\vy} \big\|_{\psi_2}.
\end{equation}
A random vector $\vx \in \R^n$ is \emph{isotropic} if it satisfies $\E \vx \vx^T = \bm{I}_n$.

The key ingredient in our proofs is the following extended matrix deviation inequality which implies the extended restricted eigenvalue condition for the sub-Gaussian sensing matrix.
\begin{proposition}[Extended Matrix deviation inequality, \cite{chen2017Matrix}] \label{lm:MatrixDeviationInequality}
 Let $\mA$ be an ${m \times n}$ random matrix whose rows are independent, centered, isotropic, and sub-Gaussian random vectors with $\max_{i} \|\mA_i\|_{\psi_2} \leq K$. For any bounded subset $\mathcal{T} \subset \R^{n} \times \R^{m}$ and $t \ge 0$, the event
  \begin{multline}
  \sup_{ (\va,\vb)\in \mathcal{T} \cap \S^{n+m-1}  }\left| \|\mA\va+\sqrt{m}\vb\|_2 - \sqrt{m} \right| \\ \leq CK^2[ \gamma(\mathcal{T}\cap \S^{n+m-1}) + t]
  \end{multline}
  holds with probability at least $1-\exp\{-t^2\}$.

\end{proposition}

\section{Main Results}
\label{main results bounded noise}

In this section, we present our main results. We use $c, C, C', C'',$ and $\epsilon$ to denote generic absolute constants.
\subsection{Recovery via Constrained Optimization}
We start with analyzing the constrained convex recovery procedures \eqref{Constrained_Optimization_I} and \eqref{Constrained_Optimization_II}.
Our first result shows that, with high probability, approximately
\begin{align}\label{NumberofMeasurements11}
CK^4\omega^2(\TT_f(\vx^{\star})\cap\S^{n-1}) + CK^4\omega^2(\TT_g(\vv^{\star})\cap\S^{m-1})
\end{align}
corrupted measurements suffice to recover $(\vx^{\star}, \vv^{\star})$ exactly in the absence of noise and stably in the presence of noise, via either of the procedures \eqref{Constrained_Optimization_I} or \eqref{Constrained_Optimization_II}.

Before stating our result, we need to define the error set
\begin{multline*}
	\EE_1(\vx^{\star},\vv^{\star}):=\{(\va,\vb)\in \R^n\times\R^m: \\ f(\vx^{\star}+\va)\leq f(\vx^{\star})  \text{ and } g(\vv^{\star}+\vb)\leq g(\vv^{\star}) \},
\end{multline*}
in which the error vector $(\hat{\vx}-\vx^{\star},\hat{\vv}-\vv^{\star})$ lives. By the convexity of $f$ and $g$, $\EE_1(\vx^{\star},\vv^{\star})$ belongs to the following convex cone
\begin{multline*}
\CC_1(\vx^{\star},\vv^{\star}):=\{(\va,\vb)\in \R^n\times\R^m: \langle \va, \vu\rangle \leq 0\\
\text{ and } \langle \vb, \vs\rangle \leq 0 \text{ for any }\vu\in\partial f(\vx^{\star})\text{ and }\vs\in \partial g(\vv^{\star}) \},
\end{multline*}
which is equivalent to
\begin{equation*}
      \{(\va,\vb)\in \R^n\times\R^m: \va \in \mathcal{T}_f(\vx^{\star})  \text{ and } \vb \in \mathcal{T}_g(\vv^{\star}) \}.
\end{equation*}
Then we have the following results.
\begin{theorem}[Constrained Recovery]
	\label{them: Constrained Recovery}
	Let $(\hat{\vx}, \hat{\vv})$ be the solution to either of the constrained optimization problems \eqref{Constrained_Optimization_I} or \eqref{Constrained_Optimization_II}. If the number of measurements
	\begin{align}\label{NumberofMeasurements1}
	\sqrt{m} \geq  CK^2 \gamma( \CC_1(\vx^{\star},\vv^{\star})\cap\S^{n+m-1} ) + \epsilon,
	\end{align}
	then
	\begin{align*}
	\sqrt{\|\hat{\vx}-\vx^{\star}\|_2^2+\|\hat{\vv}-\vv^{\star}\|_2^2}\leq \frac{2\delta\sqrt{m}}{\epsilon}
	\end{align*}
	with probability at least $1-\exp\{-\gamma^2(\CC_1\cap\S^{n+m-1})\}$.
\end{theorem}

\begin{proof}
    Since $(\hat{\vx}, \hat{\vv})$ solves \eqref{Constrained_Optimization_I} or \eqref{Constrained_Optimization_II}, we have $f(\hat{\vx}) \leq f(\vx^{\star})$ and $g(\hat{\vv}) \leq g(\vv^{\star})$. This implies $(\hat{\vx}-\vx^{\star},\hat{\vv}-\vv^{\star})\in\EE_1(\vx^{\star},\vv^{\star}) \subset  \CC_1(\vx^{\star},\vv^{\star})$.
    It then follows from Proposition \ref{lm:MatrixDeviationInequality} and \eqref{NumberofMeasurements1} that the event
	\begin{multline}\label{LowerBound1}
	\min_{(\va,\vb)\in \CC_1(\vx^{\star},\vv^{\star})\cap\S^{n+m-1} }\sqrt{m}\|\bm{\Phi}\va+\vb\|_2 \\ \geq \sqrt{m}-CK^2 \gamma( \CC_1(\vx^{\star},\vv^{\star})\cap\S^{n+m-1} ) \geq \epsilon
	\end{multline}
    holds with probability at least $1-\exp\{-\gamma^2( \CC_1(\vx^{\star},\vv^{\star})\cap\S^{n+m-1} )\}$.

    On the other hand, since both $(\hat{\vx},\hat{\vv})$ and $(\vx^{\star},\vv^{\star})$ are feasible, by triangle inequality, we have
	\begin{multline} \label{prf them1 : step2}
	\|\bm{\Phi}(\hat{\vx}-\vx^{\star})+(\hat{\vv}-\vv^{\star})\|_2 \\ \leq \|\vy-\bm{\Phi}\hat{\vx}-\hat{\vv}\|_2 + \|\vy-\bm{\Phi}\vx^{\star}-\vv^{\star}\|_2\leq 2\delta.
	\end{multline}

    Combining \eqref{LowerBound1} and \eqref{prf them1 : step2} completes the proof.
\end{proof}

To obtain interpretable sample size bound \eqref{NumberofMeasurements11} in terms of familiar parameters, it is necessary to bound $\gamma(\CC_{1}(\vx^{\star}, \vv^{\star})\cap\S^{n+m-1})$.
\begin{lemma}\label{BoundofGaussianComplexity1}
	The Gaussian complexity of $\CC_{1}(\vx^{\star}, \vv^{\star})\cap\S^{n+m-1}$ satisfies
	\begin{multline*}
	\gamma(\CC_{1}(\vx^{\star}, \vv^{\star})\cap\S^{n+m-1}) \\ \leq  2\big[\gw{\TT_f(\vx^{\star})\cap\S^{n-1}} + \gw{\TT_g(\vx^{\star})\cap\S^{m-1}} + 1\big].
	\end{multline*}
\end{lemma}

\begin{proof}
	\begin{align*}
	\gamma(& \CC_{1}(\vx^{\star}, \vv^{\star})\cap\S^{n+m-1}) \\
      &=\E\sup_{(\va,\vb)\in \CC_{1}(\vx^{\star}, \vv^{\star})\cap\S^{n+m-1} } \big|\ip{\vg}{\va} + \ip{\vh}{\vb} \big|\\
	&\leq \E\sup_{ \substack{ c\in(0,1)\\\va\in\TT_f(\vx^{\star})\cap\S^{n-1}\\\vb\in\TT_g(\vv^{\star})\cap\S^{m-1}} }c \cdot\left| \ip{\vg}{\va}\right| + \sqrt{1-c^2} \left|\ip{\vh}{\vb} \right|\\
	&\leq \E\sup_{\va\in\TT_f(\vx^{\star})\cap\S^{n-1}} \left| \ip{\vg}{\va} \right| + \E\sup_{\vb\in\TT_g(\vv^{\star})\cap\S^{m-1}} \left| \ip{\vh}{\vb} \right|\\
	&=\gamma(\TT_f(\vx^{\star})\cap\S^{n-1}) + \gamma( \TT_g(\vv^{\star})\cap\S^{m-1} )\\
    &\leq 2\big[\gw{\TT_f\cap\S^{n-1}} + \gw{\TT_g\cap\S^{m-1}} + 1\big].
	\end{align*}
The last inequality follows from \eqref{Relation}.
\end{proof}
Clearly, \eqref{NumberofMeasurements11} follows from Theorem \ref{them: Constrained Recovery} and Lemma \ref{BoundofGaussianComplexity1}.

\subsection{Recovery via Partially Penalized Optimization}
We next present performance analysis for the partially penalized optimization problem \eqref{Partially_Penalized_Optimization}. Let $\lambda = \lambda_2 / \lambda_1$, $\lambda_1$ and $\lambda_2$ are absolute constants. Our second result shows that, with high probability, approximately
\begin{align}\label{NumberofMeasurements22}
CK^4\eta^2(\lambda_1\cdot\partial f(\vx^{\star})) + CK^4\eta^2(\lambda_2\cdot\partial g(\vv^{\star}))
\end{align}
corrupted measurements suffice to recover $(\vx^{\star}, \vv^{\star})$ exactly in the absence of noise and stably in the presence of noise, via the procedure \eqref{Partially_Penalized_Optimization}.

In this case, we define the following error set
\begin{multline*}
\EE_2(\vx^{\star},\vv^{\star}):=\{(\va,\vb)\in\R^n\times\R^m: \\ f(\vx^{\star}+\va)+\lambda\cdot g(\vv^{\star}+\vb)\leq f(\vx^{\star})+\lambda \cdot g(\vv^{\star}) \}.
\end{multline*}
By the convexity of $f$ and $g$, $\EE_2(\vx^{\star},\vv^{\star})$ belongs to the following convex cone
\begin{multline*}
\CC_2(\vx^{\star},\vv^{\star}):=\{(\va,\vb)\in \R^n\times\R^m:  \langle \va, \vu \rangle
+ \lambda \langle \vb, \vs \rangle \leq 0 \\\text{ for any }\vu\in \partial f(\vx^{\star})\text{ and } \vs\in \partial g(\vv^{\star}) \}.
\end{multline*}
Then we have the following results.

\begin{theorem}[Partially Penalized Recovery]
	\label{them: Partially_Penalized_Recovery}
	Let $(\hat{\vx}, \hat{\vv})$ be the solution to the partially penalized optimization problem \eqref{Partially_Penalized_Optimization}. If the number of measurements
	\begin{align}\label{NumberofMeasurements2}
	\sqrt{m} \geq  CK^2 \gamma( \CC_2(\vx^{\star},\vv^{\star})\cap\S^{n+m-1} ) + \epsilon,
	\end{align}
	then
	\begin{align*}
	\sqrt{\|\hat{\vx}-\vx^{\star}\|_2^2+\|\hat{\vv}-\vv^{\star}\|_2^2}\leq \frac{2\delta\sqrt{m}}{\epsilon}
	\end{align*}
	with probability at least $1-\exp\{-\gamma^2(\CC_2\cap\S^{n+m-1})\}$.
\end{theorem}

\begin{proof}
The proof is similar to that of Theorem \ref{them: Constrained Recovery}.
%
%
\end{proof}

Let $\eta^2_f = \eta^2(\lambda_1\cdot\partial f(\vx^{\star}))$ and $\eta^2_g = \eta^2(\lambda_2\cdot\partial g(\vv^{\star}))$. We can bound $\gamma(\CC_{2}(\vx^{\star}, \vv^{\star})\cap\S^{n+m-1})$ as follows.
\begin{lemma}\label{BoundofGaussianComplexity2}
	The Gaussian complexity of $\CC_{2}(\vx^{\star}, \vv^{\star})\cap\S^{n+m-1}$ satisfies
	\begin{align*}
	\gamma(\CC_{2}(\vx^{\star}, \vv^{\star})\cap\S^{n+m-1}) \leq  2\sqrt{\eta^2_f + \eta^2_g} +1.
	\end{align*}
\end{lemma}

\begin{proof}
	For any point $(\va,\vb)\in\CC_2(\vx^{\star},\vv^{\star})$, we have
	\begin{align*}
	\langle \va, \vu \rangle + \lambda \langle \vb, \vs \rangle \leq 0
	\end{align*}
	for any $\vu\in\partial f(\vx^{\star})$ and $\vs\in\partial g(\vv^{\star})$.
	Multiplying both sides by $\lambda_1$ yields
	\begin{align*}
	\langle \va, \lambda_1\vu\rangle +  \langle \vb, \lambda_2\vs\rangle \leq 0.
	\end{align*}
    For any $\vg \in \R^n$ and $\vh \in \R^m$, by Cauchy-Schwarz inequality, we have
	\begin{align*}
	\ip{\va}{\vg} + \ip{\vb}{\vh}  &\leq \ip{\va}{\vg-\lambda_1 \vu} + \ip{\vb}{\vh - \lambda_2 \vs }\\
	&\leq \|\va\|_2\|\vg-\lambda_1 \vu \|_2 + \|\vb\|_2\|\vh - \lambda_2 \vs\|_2.
	\end{align*}
	Choosing suitable $\vu\in\partial f(\vx^{\star})$ and $\vs\in \partial g(\vv^{\star})$ such that
	$$\|\vg-\lambda_1\cdot \vu\|_2 = \dist(\vg,\lambda_1\cdot\partial f(\vx^{\star})) $$
	and
	$$\|\vh-\lambda_2\cdot \vs\|_2 = \dist(\vh,\lambda_2\cdot\partial g(\vv^{\star})),$$
    we obtain
    \begin{align}\label{GaussianWidthBound}
	&\ip{\va}{\vg} + \ip{\vb}{\vh} \\ \notag
    &\quad \leq \|\va\|_2\cdot \dist(\vg,\lambda_1\cdot\partial f(\vx^{\star})) + \|\vb\|_2\cdot \dist(\vh,\lambda_2\cdot\partial g(\vv^{\star}))\\ \notag
    &\quad = d_f \cdot \|\va\|_2 + d_g \cdot \|\vb\|_2,
	\end{align}
    where $d_f:= \dist(\vg,\lambda_1 \partial f(\vx^{\star}))$ and $d_g := \dist(\vh,\lambda_2 \partial g(\vv^{\star}))$.

    Therefore,
    \begin{align*}
	&\gw{{\CC_2}(\vx^{\star},\vv^{\star})\cap\S^{n+m-1}} \\
    &\quad =\E\sup_{(\va,\vb)\in \CC_2(\vx^{\star},\vv^{\star})\cap\S^{n+m-1}}\big[ \ip{\vg}{\va}+\ip{\vh}{\vb} \big] \\
	&\quad \leq \E\sup_{(\va,\vb)\in \CC_2(\vx^{\star},\vv^{\star})\cap\S^{n+m-1}}\big[ \|\va\|_2\cdot d_f + \|\vb\|_2\cdot d_g \big]\\
	&\quad \leq \E\sqrt{d_f^2 + d_g^2} \leq \sqrt{\E d_f^2 + \E d_g^2} =\sqrt{\eta^2_f + \eta^2_g}.
	\end{align*}
The second and the third inequalities follow from Cauchy-Schwarz and Jensen's inequalities respectively. By \eqref{Relation}, we complete the proof.
\end{proof}
Thus, \eqref{NumberofMeasurements22} follows from Theorem \ref{them: Partially_Penalized_Recovery} and Lemma \ref{BoundofGaussianComplexity2}.

\subsection{Recovery via Fully Penalized Optimization}
Finally, we analyze the fully penalized optimization problem \eqref{Fully Penalized Optimization}. In this case, we require regularization parameters $\tau_1$ and $\tau_2$ to satisfy the following assumption:
\begin{assumption} \label{Assump: regular}
\begin{align*}
\tau_1\geq \beta f^{\ast}(\bm{\Phi}^{T}\vz)\quad \text{and} \quad \tau_2\geq \beta g^{\ast}(\vz),
\end{align*}
for any $\beta >1$.
\end{assumption}

Our third result shows that, with high probability, approximately
\begin{align}\label{NumberofMeasurements33}
CK^4\left[\sqrt{\eta^2(\tau_1\cdot\partial f(\vx^{\star}))+\eta^2(\tau_2\cdot\partial g(\vv^{\star}))}+\frac{\tau_1\alpha_f+\tau_2\alpha_g}{\beta}\right]^2
\end{align}
corrupted measurements suffice to recover $(\vx^{\star}, \vv^{\star})$ exactly in the absence of noise and stably in the presence of noise, via the procedure \eqref{Fully Penalized Optimization}.

Similarly, define the error set
   \begin{multline*}
    \EE_3(\vx^{\star},\vv^{\star}) := \\ \{(\va,\vb)\in\R^n\times \R^m : \tau_1 f(\vx^{\star}+\va)+\tau_2 g(\vv^{\star}+\vb) \\ \leq \tau_1 f(\vx^{\star})+\tau_2  g(\vv^{\star}) + \frac{1}{\beta}[\tau_1 f(\va)+\tau_2 g(\vb)] \}.
   \end{multline*}
By the convexity of $f$ and $g$, $\EE_3(\vx^{\star},\vv^{\star})$ belongs to the following convex cone
\begin{multline*}
      \CC_3(\vx^{\star},\vv^{\star}):=\{(\va,\vb)\in \R^n\times\R^m: \\ \tau_1 \langle \va, \vu\rangle + \tau_2 \langle \vb, \vs\rangle \leq \frac{1}{\beta}[\tau_1 f(\va) + \tau_2 g(\vb)] \}
\end{multline*}
for and $\vu\in \partial f(\vx^{\star})$ and $\vs\in\partial g(\vv^{\star})$.
Then we have the following result.
\begin{theorem}[Fully Penalized Recovery]
	\label{them: Fully Penalized Recovery}
	Let $(\hat{\vx}, \hat{\vv})$ be the solution to the fully penalized optimization problem \eqref{Fully Penalized Optimization} with $\tau_1$ and $\tau_2$ satisfying Assumption \ref{Assump: regular}. If the number of
measurements
	\begin{align}\label{NumberofMeasurements3}
	\sqrt{m} \geq  CK^2 \gamma( \CC_3(\vx^{\star},\vv^{\star})\cap\S^{n+m-1} ) + \epsilon,
	\end{align}
	then
	\begin{align*}
	\sqrt{\|\hat{\vx}-\vx^{\star}\|_2^2+\|\hat{\vv}-\vv^{\star}\|_2^2}\leq 2m \cdot \frac{\beta+1}{\beta} \cdot \frac{\tau_1 \alpha_f+\tau_2 \alpha_g}{\epsilon^2}
	\end{align*}
	with probability at least $1-\exp\{-\gamma^2(\CC_3\cap\S^{n+m-1})\}$.
\end{theorem}
\begin{proof}
	Since $(\hat{\vx},\hat{\vv})$ solves \eqref{Fully Penalized Optimization}, we have
	\begin{multline}\label{OptimalCondition}
	\frac{1}{2}\|\vy-\bm{\Phi}\hat{\vx}-\hat{\vv}\|_2^2 + \tau_1 f(\hat{\vx})+\tau_2 g(\hat{\vv}) \\ \leq \frac{1}{2}\|\vy-\bm{\Phi}\vx^{\star}-\vv^{\star}\|_2^2 +\tau_1 f(\vx^{\star})+\tau_2 g(\vv^{\star}).
	\end{multline}
    Observe that
	\begin{multline*}
	\frac{1}{2}\|\vy-\bm{\Phi}\hat{\vx}-\hat{\vv}\|_2^2  = \frac{1}{2}\|\bm{\Phi}(\hat{\vx}-\vx^{\star})+(\hat{\vv}-\vv^{\star})\|_2^2\\
    +\frac{1}{2}\|\vz\|_2^2  -\ip{\bm{\Phi}(\hat{\vx}-\vx^{\star})}{\vz} -\ip{\hat{\vv}-\vv^{\star}}{\vz}.
	\end{multline*}
    Substituting this into \eqref{OptimalCondition} yields
    \begin{multline}\label{InequalityofTheorem3}
     \frac{1}{2}\|\bm{\Phi}(\hat{\vx}-\vx^{\star})+(\hat{\vv}-\vv^{\star})\|_2^2  \leq \tau_1[f(\vx^{\star})-f(\hat{\vx})]  \\
     +\tau_2 [g(\vv^{\star})-g(\hat{\vv})] + \ip{\bm{\Phi}(\hat{\vx}-\vx^{\star})}{\vz} + \ip{\hat{\vv}-\vv^{\star}}{\vz}.
    \end{multline}
    Since $\|\bm{\Phi}(\hat{\vx}-\vx^{\star})+(\hat{\vv}-\vv^{\star})\|_2^2\geq 0$, we have
    \begin{align*}
    \tau_1 &f(\hat{\vx})+\tau_2 g(\hat{\vv}) \\
                 & \leq \tau_1 f(\vx^{\star}) + \tau_2 g(\vv^{\star}) + \ip{\bm{\Phi}(\hat{\vx}-\vx^{\star})}{\vz} + \ip{\hat{\vv}-\vv^{\star}}{\vz} \\
                 & \leq \tau_1 f(\vx^{\star}) + \tau_2 g(\vv^{\star}) + f^*(\bm{\Phi}^T\vz)\cdot f(\hat{\vx}-\vx^{\star}) \\
                 &~~~~+ g^*(\vz)\cdot g(\hat{\vv}-\vv^{\star}) \\
                 & \leq \tau_1 f(\vx^{\star}) + \tau_2 g(\vv^{\star}) + \frac{\tau_1}{\beta}\cdot f(\hat{\vx}-\vx^{\star}) + \frac{\tau_2}{\beta}\cdot g(\hat{\vv}-\vv^{\star}),
    \end{align*}
    where $f^*(\cdot)$ and $g^*(\cdot)$ denotes the dual norms of $f(\cdot)$ and $g(\cdot)$ respectively. The second inequality follows from generalized H\"{o}lder's inequality. The last inequality holds because of Assumption 1. This implies $(\hat{\vx}-\vx^{\star}, \hat{\vv}-\vv^{\star}) \in \EE_3(\vx^{\star},\vv^{\star}) \subset \CC_3(\vx^{\star},\vv^{\star})$.
    It then follows from Proposition \ref{lm:MatrixDeviationInequality} and \eqref{NumberofMeasurements3} that the event
	\begin{multline}\label{LowerBound3}
	\min_{(\va,\vb)\in \CC_3(\vx^{\star},\vv^{\star})\cap\S^{n+m-1} }\sqrt{m}\|\bm{\Phi}\va+\vb\|_2 \\ \geq \sqrt{m}-CK^2 \gamma( \CC_3(\vx^{\star},\vv^{\star})\cap\S^{n+m-1} ) \geq \epsilon
	\end{multline}
    holds with probability at least $1-\exp\{-\gamma^2( \CC_3(\vx^{\star},\vv^{\star})\cap\S^{n+m-1} )\}$.

    On the other hand, it follows from \eqref{InequalityofTheorem3} that
    \begin{align}\label{UpperBound3}
    &\frac{1}{2}\|\bm{\Phi}(\hat{\vx}-\vx^{\star})+(\hat{\vv}-\vv^{\star})\|_2^2 \\ \notag
    &\quad \leq \frac{\tau_1}{\beta}\cdot f(\hat{\vx}-\vx^{\star}) + \frac{\tau_2}{\beta}\cdot g(\hat{\vv}-\vv^{\star})+\tau_1 \cdot f(\hat{\vx}-\vx^{\star})\\ \notag
    &\quad~~~+\tau_2\cdot g(\hat{\vv}-\vv^{\star})\\ \notag
	&\quad=\frac{\beta+1}{\beta}\big( \tau_1\cdot f(\hat{\vx}-\vx^{\star})+\tau_2\cdot g(\hat{\vv}-\vv^{\star}) \big)\\ \notag
	&\quad= \frac{\beta+1}{\beta}\big( \alpha_f\tau_1\cdot \|\hat{\vx}-\vx^{\star}\|_2+\alpha_g\tau_2\cdot \|\hat{\vv}-\vv^{\star}\|_2 \big)\\ \notag
	&\quad \leq \frac{\beta+1}{\beta}\cdot( \alpha_f\tau_1+ \alpha_g\tau_2 )\cdot\sqrt{\|\hat{\vx}-\vx^{\star}\|_2^2+\|\hat{\vv}-\vv^{\star}\|_2^2},
	\end{align}
    where $\alpha_f=\sup_{\vu\neq 0}\frac{f(\vu)}{\|\vu\|_2} $ and $\alpha_g=\sup_{\vu\neq 0}\frac{g(\vu)}{\|\vu\|_2}$ are compatibility constants. The first inequality follows from triangle inequality. In the last inequality, we have used Cauchy-Schwarz inequality.

    Combining \eqref{LowerBound3} and \eqref{UpperBound3} completes the proof.
\end{proof}
	
To bound the Gaussian complexity of $\CC_{3}(\vx^{\star}, \vv^{\star})\cap\S^{n+m-1}$, we have
\begin{lemma}
	\begin{align*}
	\gamma(\CC_{3}(\vx^{\star}, \vv^{\star})\cap\S^{n+m-1}) \leq  2\left[\sqrt{\eta_f^2+\eta_g^2}+\frac{\tau_1\alpha_f+\tau_2\alpha_g}{\beta}\right] + 1.
	\end{align*}
\end{lemma}

\begin{proof}
    By \eqref{GaussianWidthBound}, we obtain
	\begin{align*}
    &\gw{ \CC_3(\vx^{\star},\vv^{\star})\cap\S^{n+m-1}}\\
    &\quad =\E\sup_{(\va,\vb)\in\CC_3(\vx^{\star},\vv^{\star})\cap\S^{n+m-1}}\big[\ip{\vg}{\va}+\ip{\vh}{\vb} \big]\\
	&\quad \leq \E\big[ \|\va\|_2\cdot d_f + \|\vb\|_2\cdot d_g + \frac{1}{\beta} \tau_1\cdot f(\va)+ \frac{1}{\beta}\tau_2\cdot g(\vb) \big]\\
	&\quad \leq \sqrt{\eta_f^2+\eta_g^2}+\frac{\tau_1\alpha_f+\tau_2\alpha_g}{\beta}.
	\end{align*}
\end{proof}
The following Lemma indicates how to choose regularization parameters $\tau_1$ and $\tau_2$ in Assumption \ref{Assump: regular}.

\begin{lemma}
	\label{lm: upper bound of general ip}
Let $\mA$ be an $m \times n$ matrix whose rows $\mA_i$ are independent centered isotropic sub-Gaussian vectors with $\max_{i} \|\mA_i\|_{\psi_2} \leq K$, and $\vw $ be any fixed vector. Let $\TT$ be any bounded subset $\R^n$. Then, for any $t\geq 0$, the event
\begin{align*}
	\sup_{\vu\in \TT} \ip{\mA\vu}{\vw} \leq CK \|\vw\|_2\big[ \gamma(\TT) + t\cdot\rad(\TT) \big]
\end{align*}
holds with probability at least $1-\exp\{-t^2\}$, where $\rad(\TT):=\sup_{\vu\in\TT}\|\vu\|_2$.
\end{lemma}

\begin{proof}
Define the random process
\begin{align*}
X_{\vu}:=\ip{\mA\vu}{\vw},\text{ for any }\vu \in \TT,
\end{align*}
which has sub-Gaussian increments
\begin{align*}
\|X_{\vu}-X_{\vu'}\|_{\psi_2} & = \|\ip{\mA(\vu-\vu')}{\vw}\|_{\psi_2} \\
                              &\leq  \|\vw\|_2 \|\mA(\vu-\vu')\|_{\psi_2}\\
                              &\leq CK \|\vw\|_2 \|\vu-\vu'\|_2
\end{align*}
for any $\vu,\vu' \in \TT$. The last inequality follows from \cite[Lemma 3.4.3]{vershynin2016book}. Define $\bar{\TT} = \TT \cup \{\vzero\}$. It follows from Talagrand's Majorizing Measure Theorem \cite[Theorem 4.1]{liaw2016simple} that the event
\begin{align*}
\sup_{\vu \in \TT}\ip{\mA\vu}{\vw} &\leq \sup_{\vu \in \TT}|\ip{\mA\vu}{\vw}|= \sup_{\vu \in \bar{\TT}}|\ip{\mA\vu}{\vw}|\\
                                 & = \sup_{\vu \in \bar{\TT}}|\ip{\mA\vu}{\vw} - \ip{\mA\vzero}{\vw}|\\
                                 & \leq \sup_{\vu, \vu' \in \bar{\TT}}|\ip{\mA\vu}{\vw} - \ip{\mA\vu'}{\vw}|\\
                                 & \leq C'K\|\vw\|_2(\omega(\bar{\TT})+t\diam(\bar{\TT}))\\
                                 & \leq C''K \|\vw\|_2(\gamma(\TT)+t\rad(\TT))
\end{align*}
holds with probability at least $1-\exp\{-t^2\}$, where $\diam(\bar{\TT}):=\sup_{\vu,\vs\in\bar{\TT}}\|\vu-\vs\|_2$. In the last inequality, we have used the facts that $\omega(\bar{\TT}) \leq \gamma(\bar{\TT}) = \gamma(\TT)$ and $\diam(\TT) \leq 2\rad(\TT)$. This completes the proof.
\end{proof}

When the noise is bounded $(\|\vz\|_2\leq \delta)$, we have the event
\begin{align*}
	f^*(\bm{\Phi}^T\vz) = \sup_{\vu\in\B_f^n}\ip{\bm{\Phi}\vu}{\vz} \leq \frac{CK\delta}{\sqrt{m}}\big[ \gamma(\B_f^n)+ \sqrt{m} \cdot r_f \big]
\end{align*}
holds with probability at least $1-\exp(m)$, where $\B_f^n=\{ ~\vu\in\R^n : f(\vu)\leq 1~\}$ and $r_f=\sup\{ ~\|\vu\|_2:\vu\in\B_f^n ~ \}$. Thus it is safe to choose $\tau_1 \geq \beta\frac{CK\delta}{\sqrt{m}}\big[ \gw{\B_f^n}+\sqrt{m}\cdot r_f \big]$. In addition, we have $g^*(\vz)=\sup_{\vu\in\B_g^m}\ip{\vz}{\vu}\leq \delta \sup_{\vu\in\B_g^m}\|\vu\|_2=\delta\cdot r_g$, where $\B_g^m=\{ ~\vs\in\R^m : g(\vs)\leq 1~\}$ and $r_g=\sup\{ ~\|\vs\|_2:\vs\in\B_g^m ~ \}$. Therefore, we can choose $\tau_2 \geq \beta \delta\cdot r_g$.

When $\vz$ is a sub-Gaussian random vector such that \eqref{model: noise} holds, then $\|\vz\|_2$ concentrates near the value $\sqrt{m}$ \cite[Theorem 3.1.1]{vershynin2016book}, that is $\|\|\vz\|_2 - \sqrt{m}\|_{\psi_2} \leq CK^2$. This implies
\begin{multline*}
  \Pr{\|\vz\|_2 \geq (L^2+1)\sqrt{m}} \\ \leq \Pr{\big|\|\vz\|_2 - \sqrt{m}\big| \geq L^2\sqrt{m} } \leq 2 e^{-cm}.
\end{multline*}
Combining this with Lemma \ref{lm: upper bound of general ip} and taking union bound yields
\begin{align*}
	f^*(\bm{\Phi}^T\vz) = \sup_{\vu\in\B_f^n}\ip{\bm{\Phi}\vu}{\vz} \leq {CK(1+L^2)}\big[ \gamma(\B_f^n)+ \sqrt{m} \cdot r_f \big]
\end{align*}
with probability at least $1-3e^{-cm}$. Moreover, it is not hard to show the event
\begin{align*}
  g^*(\vz)=\sup_{\vu\in\B_g^m}\ip{\vz}{\vu}\leq {CL}\big[ \gamma(\B_g^m)+ \sqrt{m} \cdot r_g \big]
\end{align*}
holds with probability at least $1-\exp\{-m\}$. In order to satisfy the Assuption \ref{Assump: regular}, we can choose $\tau_1 \geq {CK(1+L^2)\beta}\big[ \gamma(\B_f^n)+ \sqrt{m} \cdot r_f \big]$ and $\tau_2 \geq {CL\beta}\big[ \gamma(\B_g^m)+ \sqrt{m} \cdot r_g \big]$ in the sub-Gaussian noise case.

\section{Conclusion}
In this paper, we have presented performance analysis for three convex recovery procedures which are used to recover a structured signal from corrupted sub-Gaussian measurements. We considered both bounded and stochastic noise cases. Our results have shown that, under mild conditions, these approaches reconstruct both signal and corruption exactly in the absence of noise and stably in the presence of noise. For future work, it would be of great interest to exploit the relationship among these procedures and their phase transition phenomenon \cite{Zhang2017}.

\bibliographystyle{IEEEtran}
\bibliography{IEEEabrv,myref}


\end{document}